\documentclass[preprint,12pt]{elsarticle}

\usepackage{amsfonts, amsthm, amsgen, amssymb}
\usepackage{latexsym,amsmath}
\usepackage{array}
\usepackage{hyperref}

\theoremstyle{plain}
\newtheorem{theorem}{Theorem}
\newtheorem{lemma}[theorem]{Lemma}
\theoremstyle{definition}
\newtheorem{definition}[theorem]{Definition}
\newtheorem{remark}[theorem]{Remark}
\newtheorem{question}[theorem]{Question}

\newcommand{\Z}{\mathbb{Z}}
\newcommand{\Q}{\mathbb{Q}}

\newcommand{\F}{\mathbb{F}}
\newcommand{\Aut}{\mathrm{Aut}}
\newcommand{\wt}{\mathrm{wt}}
\newcommand{\supp}{\mathrm{supp}}

\journal{Finite Fields and their Applications}

\begin{document}

\begin{frontmatter}

\title{Some new results on the self-dual [120,60,24] code}

\author{Martino~Borello}

\address{Universit\'e Paris 13, Sorbonne Paris Cit\'e, LAGA, CNRS, UMR 7539, Universit\'e Paris 8, F-93430, Villetaneuse,
France}

\author{Javier~de~la~Cruz}

\address{Universidad del Norte, Barranquilla, Colombia and University
of Zurich, Switzerland}

\begin{abstract}
The existence of an extremal self-dual binary linear code of length
120 is a long-standing open problem. We continue the investigation
of its automorphism group, proving that automorphisms of order 30
and 57 cannot occur. Supposing the involutions acting fixed point
freely, we show that also automorphisms of order 8 cannot occur and
the automorphism group is of order at most 120, with further
restrictions. Finally, we present some necessary conditions for the
existence of the code, based on shadow and design theory.
\end{abstract}


\end{frontmatter}

\section{Introduction}
In coding theory, binary self-dual codes play a central role: they are linear codes with a rich algebraic structure, good decoding properties and relations with other areas of mathematics, such as group theory, lattice theory and design theory. For example, this class includes the binary extended Golay code, whose automorphism group is the sporadic simple group $\rm M_{24}$ and which is related to the Leech lattice.\\
Gleason, Pierce and Turyn  showed (see \cite{Turyn}) that if a
natural number $r>1$ divides the weight of all codewords of a binary
self-dual code, then $r=2$ (\textit{even} code) or $r=4$
(\textit{doubly-even} code). Every binary self-dual code is even. If
a binary self-dual code is even but not doubly-even
(\textit{singly-even} code), then it is called a \textit{Type I}
code, while if a binary self-dual code is doubly-even, then it is
called a \textit{Type II} code. Type II codes exist only for lengths
which are multiples of $8$ \cite{Gleason} and Mallows and Sloane
showed in \cite{Mallows} that they have minimum distance bounded by
$4\lfloor{n/24}\rfloor+4$, where $n$ is the length. A type II code
attaining this bound is called \textit{extremal} code. Among
extremal codes, those of length a multiple of $24$ are particularly
interesting: Assmus-Mattson's theorem \cite{Assmus-paper} guarantees
that the supports of their codewords of a fixed nonzero weight form
a $5$-design. Moreover, they have relations, as mentioned above,
with simple groups and extremal lattices. Zhang proved in
\cite{Zhang} that their length is at most $3672$.

Despite their theoretical importance, only two extremal codes of
length a multiple of $24$ are known, namely the famous binary
extended Golay code, the unique up to equivalence of length $24$,
and the extended quadratic residue code of length $48$, which is the
unique up to equivalence of this length. In 1973 Sloane
\cite{pregunta} posed explicitly the question: is there a self-dual
$[72, 36, 16]$ code? Since then, multiple attempts to establish the
non existence of such a code or to present a construction have been
done, till now unsuccessfully. The problem is still open for all
lengths from $72$ to $3672$ and many investigations have been also
done for the cases of length $96$ and $120$.

This paper focuses on the last one, i.e. on the study of a self-dual
$[120, 60, 24]$ code. In particular, in Section \ref{background} we
will collect, for the reader's convenience, all the definitions and
the known results which will be used in the following. In Section
\ref{section-autom} we prove new properties about the automorphism
group of a self-dual $[120,60,24]$ code. In particular we exclude
the existence of automorphisms of order $30$ and $57$ and we
investigate the structure of the automorphism group, in the case
that involutions act fixed point freely (see the introduction of
Subsection \ref{ssfpf} for a motivation of this choice), proving
that it is either trivial or isomorphic to a group of order at most
$120$, with further restrictions. Finally, in Section
\ref{section-necess} we give necessary conditions for the existence
of the code, based on shadow and design theory.

\section{Background}\label{background}

In this section we collect some classical results of coding theory
which are useful in the rest of the paper.

\subsection{Gleason's theorem and the shadow of a code}

For the whole subsection, let $C$ be a binary code of length $n$,
i.e. a subspace of $\F_2^n$. We recall that a $[n,k,d]$ code is a
code of length $n$, dimension $k$ and minimum distance $d$.

\begin{definition}
The \textit{weight distribution} of $C$ is the sequence $(A_0(C),$ $\ldots,$ $A_n(C))$, where $A_i(C)$ is the number of codewords of $C$ of weight $i$, for every $i\in \{1,\ldots,n\}$.\\
The polynomial $W_{C}(y):=\sum_{i=0}^n A_iy^{i}=\sum_{c\in C}
y^{\wt(c)}\in \Z[x]$ is called the \textit{weight enumerator} of $C$
and the polynomial $W_C(x,y):=x^nW(\tfrac{y}{x})\in\Z[x,y],$ is the
\textit{homogeneous weight enumerator} of $C$.
\end{definition}

\begin{definition}
The \textit{dual} of $C$ is $C^\perp:=\{v\in \F_2^n\ | \ \langle v,c\rangle=0, \ \forall c\in C\}$.\\
If $C=C^\perp$, we say that $C$ is self-dual.\\
If $C$ and $C^\perp$ have the same weight enumerator, $C$ is called
a \textit{formally self-dual} code.
\end{definition}

\begin{theorem}[\cite{Gleason}]\label{Gleason-theorem}
Let $g_1(x,y):=y^2+x^2$, $g_2(x, y):=x^{2}y^{2}(x^{2}-y^{2})^{2}$,
and $g_3(x,
y):=y^{24}+759x^8y^{16}+2576x^{12}y^{12}+759x^{16}y^8+x^{24}$.
\begin{itemize}
  \item[{\rm (a)}]  If $C$ is formally self-dual and even, $$W_C(x,y)=\sum_{i=0}^{\lfloor n /8 \rfloor}a_ig_1(x,y)^{\frac{n}{2}-4i}g_2(x,y)^i.$$
  \item[{\rm (b)}]  If $C$ is formally self-dual and doubly-even, $$W_C(x,y)=\sum_{i=0}^{\lfloor n /24 \rfloor}a_ig_2(x,y)^{\frac{n}{8}-3i}g_3(x,y)^i.$$
\end{itemize}
In all cases, every $a_i\in \Q$ and $\sum_ia_i=1.$
\end{theorem}

Let $C$ be a self-dual code and let $C_{0}$ be the subset consisting
of all codewords in $C$ whose weights are multiples of $4$. If $C$
is of type II then $C_0=C$, while $C_0$ is a subcode of index $2$ of
$C$ if $C$ is of type I.

\begin{definition}  The \textit{shadow} of $C$ is the set
$$S:=\left\{%
\begin{array}{ll}
    C^{\perp}_{0}\backslash C, & \hbox{if $C$ is of type I} \\
    C, & \hbox{if $C$ is of type II.} \\
\end{array}%
\right.    $$
\end{definition}

Let $C$ be a type I code. Since $C_0$ is of index $2$, then
$\#(C_0^\perp/C_0)=4$. Hence there are three cosets $C_1,  C_2, C_3$
of $C_0$ in $C_0^\perp$ such that $C_{0}^{\perp}= C_0 \cup C_1 \cup
C_2 \cup C_3$, where $C = C_0 \cup C_2$ and $S = C_1 \cup C_3 =
C_0^{\perp} \backslash C$ is the shadow of $C$ (see \cite[Theorem
5]{conway-sloane}).

\begin{theorem}[\cite{conway-sloane}]\label{polynomial-type-I}  Let $S$ be the shadow of $C$, code of type I.
\begin{itemize}
  \item[{\rm (a)}] If we write $$W_{C}(x,y)=\sum^{\lfloor
n/8\rfloor}_{j=0} a_{j}(x^{2}+y^{2})^{ \frac{n}{2}
-4j}(x^{2}y^{2}(x^{2}-y^{2})^{2})^{j},$$ for suitable rationals
$a_j$, then $$W_{S}(x,y)=\sum^{\lfloor
n/8\rfloor}_{j=0}(-1)^{j}a_{j}2^{\frac{n}{2}-6j}(xy)^{\frac{n}{2}-4j}(x^{4}-y^{4})^{2j}.$$
  \item[{\rm (b)}] Writing $W_S(x,y)=\sum_{i=0}^{n}B_{i}x^{n-i}y^{i}$, we have
\begin{itemize}
          \item[{\rm (i)}] $B_{i}=B_{n-i}$ for all $i.$
          \item[{\rm (ii)}] $B_{i}=0,$ unless $i\equiv n/2 \mod 4.$
       \item[{\rm (iii)}] $B_{0}=0.$
        \item[{\rm (iv)}] $B_{i}\leq1,$ for $i<d/2.$
         \item[{\rm (v)}] at most one $B_{i}$ is nonzero for $i < (d+4) / 2$.
    \end{itemize}
\end{itemize}
\end{theorem}

\begin{definition}
If $C$ is a self-dual $[n, n/2, d]$ code with $d > 2$, pick two
positions and consider the $(n/2-1)$-dimensional subcode
$C^{\prime}$ of $C$ with either two $0$s or two $1$s in these
positions.  If we puncture $C^{\prime}$ on these positions, we
obtain a self-dual code $C^{\prime \ast}$ of length $n-2$;
$C^{\prime \ast}$ is called a \textit{child} of $C$ and $C$ is
called a \textit{parent} of $C^{\prime \ast}$.
\end{definition}

\begin{theorem}[\cite{Pless}] \label{child} Let $m \geq1$ be an integer.
  If $C$ is a $[24m-2, 12m-1, 4m+2]$ type I code whose shadow has minimum distance $4m+3$, then $C$ is a child of a $[24m, 12m, 4m+4]$ type II code.
  \end{theorem}

\begin{lemma} [\cite{Pless}] \label{same-polynomial} If $C$ is a child of an extremal type II code with shadow $S=C_{1}\cup C_{3}$, then
 $W_{C_{1}}(y)=W_{C_{3}}(y)$.
\end{lemma}

\begin{lemma} [\cite{formally-dual}]  \label{prop1} Let $C$ be a type I code of length $n$ with the shadow $S=C_{1}\cup C_{3}$.
Suppose that $n\equiv 2 \mod 4$. Let $C^{*}$ be the code of length
$n+2$ obtained by extending $C_{0}^{\perp}$ as follows:
  $$(0,0,C_{0}),(1,0,C_{2}),(0,1,C_{1}),(1,1,C_{3}).$$
    If $W_{C_{1}}(y)=W_{C_{3}}(y)$, then $C^{*}$ is a formally self-dual code with weight enumerator    $$W_{C_{0}}(y)+y(W_{C_{1}}(y)+W_{C_{2}}(y))+y^{2}W_{C_{3}}(y).$$
\end{lemma}

\begin{definition}
Two self-dual codes of length $n$ are  \textit{neighbors} if their
intersection is a code of dimension $n/2-1$.
\end{definition}

\subsection{Automorphism group of binary codes}
The symmetric group $\textrm{S}_n$ acts on $\F_2^n$ by the group
action $v\sigma:=(v_{\sigma^{-1}(1)}, \ldots,$
$v_{\sigma^{-1}(n)})$, where $v = (v_1, \ldots, v_n) \in \F_2^n$ and
$\sigma \in \textrm{S}_n$.

\begin{definition}\label{equivalent}
 Let $C$ and $C^{ \prime}$ be two codes of the same length $n$. We say that $C$ and $C^{\prime}$ are \textit{equivalent} and denote $C \sim C'$ if only if $C \sigma = C'$ where $\sigma\in \textrm{S}_n $. If $v \sigma \in C$ for all $v\in C$, then $\sigma$ is an \textit{automorphism} of $C$. The set of all automorphisms of $C$ is a group, denoted $\Aut(C)$.
\end{definition}

\begin{definition}
Let $C$ be a binary code of length $n$ and $\sigma\in \Aut(C)$.
\begin{itemize}
  \item[{\rm (a)}] If $\sigma$ is of prime order $p$, we say that $\sigma$ is of \textit{type} $p$-$(c;f)$ if it has $c$ cycles of length $p$ and $f$ fixed points. \item[{\rm (b)}] If $\sigma$ is of order $p\cdot r$, where $p,r$ are distinct primes, then we say that $\sigma$ is of \textit{type} $p \cdot r$-$(s_1, s_2, s_3;f)$ if $\sigma$ has $s_1$ $p$-cycles, $s_2$ $r$-cycles, $s_3$ $pr$-cycles and $f$ fixed points.
\end{itemize}
\end{definition}

\begin{remark} In order to simplify the notation, if $\sigma$ is an automorphism of composite order $r$ and has $c$ $r$-cycles and $f$ fixed points with $n=c\cdot r+f$, then we say that the cycle structure of $\sigma$ is $r$-$(c;f)$.

\end{remark}

Let us first prove a result which is useful in the following
sections.

\begin{lemma}\label{lemmaorder}
Let $C$ be a code of length $n$, such that all automorphisms of
prime order $p$ act fixed point freely. If $|\Aut(C)|=p^am$, with
$(p,m)=1$, then $a \leq \max \{ \ r \in \mathbb{Z} \ : \ p^r \mid n
\ \}$.
\end{lemma}

\begin{proof} Suppose $a > \max \{ \ r \in \mathbb{Z} \ : \ p^r \mid n \ \}$. By Sylow's theorem, there exists a subgroup $H \leq \Aut(C)$ with $|H|=p^a$. The group $H$ acts on the set $\{1,\ldots, n\}$. Since all automorphisms of order $p$ act fixed point freely, then each orbit has $p^a$ elements. Therefore $p^a \mid n$, a contradiction.
\end{proof}

\begin{definition}
Let $\sigma\in \Aut(C)$. The \textit{fixed code} of $\sigma$ is
$$F_{\sigma}(C):=\{v\in C\mid v\sigma=v\}.$$
Let $\Omega_1,\ldots,\Omega_{c}$ be the cycle sets and let
$\Omega_{c+1},\ldots,\Omega_{c+f}$ be the fixed points of $\sigma$.
Clearly $v \in F_{\sigma}(C)$ if and only if $v \in C$ and $v$ is
constant on each cycle. Let $\pi_\sigma:F_{\sigma}(C)\rightarrow
\F_2^{c+f}$ denotes the \textit{projection map} defined by
$\pi_\sigma(v|_{\Omega_{i}})=v_{j}$ for some $j\in \Omega_{i}$ and
$i\{1,\ldots,c+f\}$.
\end{definition}

A useful result, which is a reformulation of a very classical result
about group actions, is the following.

\begin{lemma}\label{coefficients}
If $\sigma\in \Aut(C)$, $W_C(y)=\sum A_iy^i$ and
$W_{F_\sigma(C)}(y)=\sum A_i^Fy^i$, then $A_{i}\equiv A_{i}^F$ {\rm
mod} $p$.
\end{lemma}

Finally, let us introduce a classical decomposition of a code with
an automorphism of prime order, which comes from Maschke's theorem.
Let $p$ be an odd prime and $\sigma$ is an automorphism of type
$p$-$(c,f)$. Let
$$ E_{\sigma}(C):=\{v\in C\mid \wt(v|_{\Omega_{i}})\equiv 0\mod 2, \ \ i=1,\ldots,c+f\},$$
where $v|_{\Omega_{i}}$ is the restriction of $v$ on $\Omega_{i}$.

\begin{lemma}[\cite{Huffman}]
If $p$ is odd, then $C=F_{\sigma}(C)\oplus E_{\sigma}(C)$. Moreover,
if $C$ is self-dual, then
\begin{enumerate}
\item[{\rm (a)}] the code $\pi_\sigma(F_{\sigma}(C))\leq \F_2^{c+f}$ is self-dual and, if $C$ is doubly even and $p \equiv 1 \mod p$, then $\pi_\sigma(F_{\sigma}(C))$ is doubly even.
\item[{\rm (b)}] $\dim \,
E_{\sigma}(C)=\frac{(p-1)c}{2}$.
\end{enumerate}
\end{lemma}

\subsection{Designs and codes}
In this section we briefly recall the main definitions of design
theory and its relationship with coding theory.

\begin{definition}
A $t$-$(v, k, \lambda)$ design, or briefly a \textit{$t$-design}, is
a pair $\mathcal{D}= (\mathcal{P}, \mathcal{B})$ where $\mathcal{P}$
is a set of $v$ elements, called \textit{points}, and $\mathcal{B}$
is a collection of distinct subsets of $\mathcal{P}$ of size $k$,
called \textit{blocks}, such that every subset of points of size $t$
is contained in precisely $\lambda$ blocks.
\end{definition}

If $\mathcal{D}$ is a $t$-$(v, k, \lambda)$ design, it is also an
$i$-$(v, k, \lambda_{i})$ design for all $i\in\{0,\ldots, t\}$,
where $\lambda_{i}$ is given by $\lambda_{i}=\lambda    \frac{\binom
{v-i} {t-i}}{\binom {k-i} {t-i}}.$

\begin{definition}
Let $\mathcal{D}=( \mathcal{P}, \mathcal{B})$ be a design with
$|\mathcal{P}|=v$ and $|\mathcal{B}|=b$.
\begin{itemize}
  \item[(a)]
 If we list the points $\{p_{1},p_{2}, \ldots ,p_{v}\}$ and the blocks $\{B_{1}, B_{2}, \ldots, B_{b}\}$, then we define the incidence matrix of $\mathcal{D}$ as a $b\times v$ matrix $A=(a_{ij})$, where
$$a_{ij}=
\left\{
  \begin{array}{ll}
    1, & \hbox{if $p_{j} \in B_{i}$} \\
    0, & \hbox{if $p_{j}\not \in B_{i}$.}
  \end{array}
\right.$$
  \item[(b)] The code $C(\mathcal{D})$ over $\mathbb{F}_{2}$ which is generated by the rows of $A$ is called \textit{code of the design} $\mathcal{D}=( \mathcal{P}, \mathcal{B})$.
\end{itemize}
\end{definition}

Although there are several incidence matrices, depending on the
choice of the order of the points and of the blocks, for a given
design the codes generated by these matrices are equivalent:
changing the order of the points is equivalent to permute the
coordinates, while a reordering of the blocks does not change the
code (see \cite[p.41]{AssmusKey}).

The following theorem, due to Assmus and Mattson, establishes a
relationship between coding theory and design theory.

\begin{theorem}[\cite{Assmus-paper}]\label{Assmus} Let $C$ be a binary $[n, k, d]$ code. Suppose $C^\perp$ has minimum weight $d^\perp$. Suppose that $A_0,\ldots,A_n$ and $A_0^{\perp},\ldots,A_n^\perp$ are the weight distributions of $C$ and $C^\perp$, respectively. Fix a positive integer $t$ with $t < d$, and let $s$ be the number of $i$ with $A_i^{\perp} \neq 0$ for $i\in \{0\ldots n-t\}$. Suppose $s \leq d-t$.
\begin{itemize}
  \item[{\rm (a)}] The vectors of weight $i$ in $C$ form a $t$-design provided $A_i\neq 0$ and $d \leq i\leq n$.
  \item[{\rm (b)}] The vectors of weight $i$ in $C^{\perp}$ form a $t$-design provided $A_i^{\perp}$ and $d^{\perp}\leq i\leq n-t.$
\end{itemize}
\end{theorem}

\section{The automorphism group of an extremal [120, 60, 24] code}\label{section-autom}

For the whole section, let $C$ be an extremal $[120,60,24]$ code.
By Theorem \ref{Gleason-theorem} (b), we can easily deduce (see
\cite{Mallows}) that
\begin{equation} \label{enumerator} W_{C}(1,y)=1+39703755y^{24}+6101289120y^{28}+475644139425y^{32}+\ldots\end{equation}
Knowledge of the existence of a non-trivial automorphism group $G$
is very useful in constructing the code, since in this case the code
has the structure of a $\F_2G$-module. For this reason, there is an
intensive research on the automorphism group of extremal codes.

\begin{remark}\label{stateofart}
Concerning the code of length $120$, the following results on the
automorphism group $G$ of $C$ are known (see \cite{Bor-W-2p,
Bouy-24m, B-DC-W, C, CKW}):
\begin{itemize}
\item[\rm (a)] The order of $G$ divides $2^a\cdot3\cdot5\cdot7\cdot19\cdot23$ for a non-negative integer $a$ (which is at most $116$, since $G\subseteq S_{120}$).
\item[\rm (b)] If $\sigma$ is an automorphism of $C$ of prime order $p$ then its cycle structure is
\begin{center} {
\begin{tabular}{c|c|c}
$p$ & \mbox{\rm number of } & \mbox{ \rm number of } \\
 & \mbox{\rm  $p$-cycles} & \mbox{ \rm fixed points} \\
\hline
$2$ & $48,60$ & $24,0$ \\
$3$ & $40$ & $0 $\\
$5$ & $24$ & $0$ \\
$7$ & $17$ & $1$ \\
$19$ & $6$ & $6$ \\
$23$ & $5$ & $5$ \\
\end{tabular} }
\end{center}
\item[\rm (c)] If $\sigma$ is an automorphism of $C$ of odd composite order $r$, then the cycle structure of $\sigma$ is either $15$-$(8;0)$, $3\cdot 19$-$(2,0,2;0)$ or $5\cdot 23$-$(1,0,1;0)$.
\end{itemize}

Moreover, \textit{if all involutions act fixed point freely}, the
following conditions hold:

\begin{itemize}

\item[\rm (d)] If $\sigma$ is an automorphism of $C$ of even composite order not divisible by $8$, then the cycle structure of $\sigma$ is either $4$-$(30;0)$, $6$-$(20;0)$, $10$-$(12;0)$, $12$-$(10;0)$, $20$-$(6;0)$, $30$-$(4;0)$, $60$-$(2;0)$.

\item[\rm (e)] The order of $G$ is in $\{7, 19, 23, 38, 56, 57, 114, 115, 552, 2760\}$ or $G$ is a $\{2,3,5\}$-group of order dividing $120$.
\end{itemize}
\end{remark}

\begin{remark}
Condition (d) is not stated explicitly in any of the above
references, but it is an easy consequence of the results in
\cite{Bor-W-2p}. Furthermore, we give only the structure of the
automorphisms of even order not divisible by $8$ because we prove in
the following that an automorphism of order $8$ cannot exist under
the hypothesis that involutions act fixed point freely.
\end{remark}

\begin{remark}
Condition (e) corrects a mistake in Proposition 15 b) of
\cite{B-DC-W}, where ``$|G|=2^3\cdot 5^c\cdot 23$'' should have been
``$|G|=2^3\cdot 3\cdot  5^c\cdot 23$''. Moreover, it gives a
preciser statement about $\{2,3,5\}$-groups, based on Lemma
\ref{lemmaorder}.
\end{remark}

\subsection{Fixed code of automorphism of prime order}

In this subsection we present some preliminary results about the
automorphisms of prime order. It is a hard problem to prove that the
primes $3,5, 7, 19$ and $23$ cannot occur as orders of an
automorphism $\sigma$ of $C$: even if we can completely determine
the fixed code $F_\sigma(C)$, there are too many possibilities to
check for the complement $E_\sigma(C)$ defined in Section
\ref{background}. Also the case of the prime $2$ is computationally
hard and we do not even know the fixed code.

\vspace{3mm}

\textbf{Automorphism of order 2:} Let $ \sigma \in \Aut(C)$ be of
order $2$. Then $\sigma$ is either of type $2$-$(48;24)$ or of type
$2$-$(60;0)$. In the second case, by Theorem 1.2 of
\cite{Bor-Nebe-Involution},  $\pi_\sigma(F_{\sigma}(C))$ is a
self-dual $[60,30,12]$ code. Although some self-dual codes with
these parameters are known, a complete classification is still
unknown.

\vspace{3mm}

\textbf{Automorphism of order 3:} Let $ \sigma \in \Aut(C)$ be of
order $3$. Then $\sigma$ is of type $3$-$(40;0)$ and $\pi_\sigma
(F_{\sigma}(C))$ is a self-dual
 doubly-even $[40,20,8]$ code. By \cite{sd40}, there are $16470$ such codes up to equivalence.

\vspace{3mm}

\textbf{Automorphism of order 5:}  Let $ \sigma \in \Aut(C)$ be of
order $5$. Then $ \sigma$ is of type $5$-$(24;0)$ and $\pi_\sigma
(F_{\sigma}(C))$ is a self-dual $[24,12,8]$ code. This implies that
$\pi_\sigma(F_{\sigma}(C))$ is equivalent to the binary extended
Golay code $G_{24}$.

\vspace{3mm}

\textbf{Automorphism of order 7:}
 Let $ \sigma \in \Aut(C)$ be of order $7$. Then $\sigma$ is of type $7$-$(17;1)$ and $\pi_\sigma (F_{\sigma}(C))$ is a self-dual $[18,9,4]$ code.
 By \cite{clasifi}, $\pi_\sigma
 (F_{\sigma}(C))$ is a equivalent to $H_{18}$ or $I_{18}$.

A vector of weight $4$ in $\pi_\sigma(F_{\sigma}(C))$ has to be a
vector of weight $28$ in $F_{\sigma}(C)$, i.e. all nonzero
coordinates of vectors of weight $4$ correspond to cycles. By the
study of clusters (see \cite{Huffman}) we can easily prove that
$H_{18}$ cannot occur. Moreover, with the same technique, we can
prove that, up to equivalence,
$${\rm gen}(F_{\sigma}(C))=\left(
                \begin{array}{ccccccccccccccccc|c}
                  \textbf{1} & \textbf{1} & \textbf{1} & \textbf{1} & \textbf{0} &\textbf{0} & \textbf{0} & \textbf{0} & \textbf{0} &\textbf{0} & \textbf{0} & \textbf{0} & \textbf{0} &\textbf{0} & \textbf{0} & \textbf{0} & \textbf{0} & 0 \\
                  \textbf{0} &\textbf{0} & \textbf{1} & \textbf{1}& \textbf{1} & \textbf{1} & \textbf{0} &\textbf{0} & \textbf{0} & \textbf{0} & \textbf{0} &\textbf{0} & \textbf{0} & \textbf{0} & \textbf{0} & \textbf{0} & \textbf{0} & 0 \\
                  \textbf{0} &\textbf{0} & \textbf{0} & \textbf{0} & \textbf{1}& \textbf{1} & \textbf{1}& \textbf{1} & \textbf{0} &\textbf{0} & \textbf{0} & \textbf{0} & \textbf{0} &\textbf{0} & \textbf{0} & \textbf{0} & \textbf{0} & 0 \\
                  \textbf{0} &\textbf{0} & \textbf{0} & \textbf{0} & \textbf{0} & \textbf{0} & \textbf{1} &  \textbf{1}& \textbf{1} & \textbf{1} & \textbf{0} &\textbf{0} & \textbf{0} & \textbf{0} & \textbf{0} &\textbf{0} & \textbf{0} & 0\\
                  \textbf{1} &\textbf{0} & \textbf{1} & \textbf{0} & \textbf{1} & \textbf{0} & \textbf{1} & \textbf{0} & \textbf{1} & \textbf{0} &\textbf{0} & \textbf{0} & \textbf{0} & \textbf{0} &\textbf{0} & \textbf{0} & \textbf{0} & 1 \\
                  \textbf{0} &\textbf{0} & \textbf{0} & \textbf{0} & \textbf{0} &\textbf{0} & \textbf{0} & \textbf{0} & \textbf{0} & \textbf{0} & \textbf{1}& \textbf{1} & \textbf{1}& \textbf{1} & \textbf{0} &\textbf{0} & \textbf{0} & 0 \\
                  \textbf{0} &\textbf{0} & \textbf{0} & \textbf{0} & \textbf{0} &\textbf{0} & \textbf{0} & \textbf{0} & \textbf{0} &\textbf{0} & \textbf{0} & \textbf{0} & \textbf{1} & \textbf{1}& \textbf{1} & \textbf{1} & \textbf{0} & 0 \\
                  \textbf{0} &\textbf{0} & \textbf{0} & \textbf{0} &  \textbf{0} &\textbf{0} & \textbf{0} & \textbf{0} & \textbf{0} & \textbf{0} & \textbf{1} & \textbf{0} & \textbf{1} & \textbf{0} & \textbf{1} & \textbf{0} & \textbf{1} & 0 \\
                  \textbf{1} & \textbf{1}& \textbf{1} & \textbf{1} & \textbf{1} & \textbf{1}& \textbf{1} & \textbf{1} & \textbf{1} & \textbf{1}& \textbf{1} & \textbf{1} & \textbf{1} & \textbf{1}& \textbf{1} & \textbf{1} & \textbf{1} & 1 \\
                \end{array}
              \right)$$
where \textbf{1} is the all-one vector and \textbf{0} the
zero-vector of length $7$.

\vspace{3mm}

\textbf{Automorphism of order 19:} Let $ \sigma \in \Aut(C)$ be of
order 19. Then $\sigma$ is of type $19$-$(6;6)$ and $\pi_\sigma
(F_{\sigma}(C))$ is a self-dual $[12,6,4]$ code. By \cite{clasifi},
$\pi_\sigma(F_{\sigma}(C))$ is equivalent to $B_{12}$.

By Lemma \ref{coefficients} and by \eqref{enumerator},
$A_{24}^F\equiv 6 \mod 19$. Therefore there are $6 \mod 19$ vectors
of $F_{\sigma}(C)$ of weight $24$. If $v\in F_{\sigma}(C)$ has
weight $24$, then $\wt(\pi_\sigma(v))=6$. Suppose that $v_{1},
v_{2}\in F_{\sigma}(C)$ of weight $24$ coincide in the coordinate
corresponding to a cycle of length $19$. Then $\wt(v_{1}+v_{2})\leq
2$. Therefore $v_{1}=v_{2}$ and there are exactly 6 vectors in
$F_{\sigma}(C)$ of weight $\wt(v)=24$. These vectors are linearly
independent and so, up to a permutation of the last six columns,
$$\label{genf}{\rm{gen}}(F_{\sigma}(C))=\left(%
\begin{array}{cccccc|cccccc}
  \mathbf{1} & \mathbf{0} & \mathbf{0} & \mathbf{0} & \mathbf{0} & \mathbf{0} & 0 & 1 & 1 & 1 & 1 &1 \\
  \mathbf{0}&  \mathbf{1} & \mathbf{0} & \mathbf{0} & \mathbf{0} & \mathbf{0} & 1 & 0 & 1 & 1 & 1 &1 \\
  \mathbf{0} & \mathbf{0} & \mathbf{1} & \mathbf{0} & \mathbf{0} & \mathbf{0} & 1 & 1 & 0 & 1 & 1 &1 \\
  \mathbf{0 }& \mathbf{0} & \mathbf{0} & \mathbf{1} & \mathbf{0} & \mathbf{0} & 1 & 1 & 1 & 0 & 1 &1 \\
  \mathbf{0} & \mathbf{0} & \mathbf{0} & \mathbf{0} & \mathbf{1} & \mathbf{0} & 1 & 1 & 1 & 1 & 0 &1 \\
  \mathbf{0} & \mathbf{0} & \mathbf{0} & \mathbf{0} & \mathbf{0} & \mathbf{1} & 1 & 1 & 1 & 1 & 1 &0 \\
\end{array}%
\right),$$\\
where \textbf{1} is the all-one vector and \textbf{0} the
zero-vector of length $19$.

\vspace{3mm}

\textbf{Automorphism of order 23:} Let $ \sigma \in \Aut(C)$ be of
order 23. Then $\sigma$ is of type $23$-$(5;5)$ and $\pi_\sigma
(F_{\sigma}(C))$ is a self-dual $[10,5,2]$ code. So (see
\cite{Wasserman}), up to equivance,
$$
{\rm gen}(F_{\sigma}(C))=\left(%
\begin{array}{ccccc|ccccc}
  \mathbf{1} & \mathbf{0} & \mathbf{0} & \mathbf{0} & \mathbf{0} & 1 & 0& 0 & 0& 0\\
  \mathbf{0} & \mathbf{1} & \mathbf{0} & \mathbf{0} & \mathbf{0} & 0 & 1& 0& 0& 0\\
  \mathbf{0} & \mathbf{0} & \mathbf{1} & \mathbf{0} & \mathbf{0} & 0 & 0& 1& 0& 0\\
  \mathbf{0} & \mathbf{0} & \mathbf{0} & \mathbf{1} & \mathbf{0} & 0 & 0& 0& 1& 0\\
  \mathbf{0} & \mathbf{0} & \mathbf{0} & \mathbf{0} & \mathbf{1} & 0 & 0& 0& 0& 1\\
\end{array}%
\right) .$$

\subsection{Automorphisms of composite order}

In this subsection we present some new results about automorphisms
of composite order. The result for the automorphism of order $8$ is
a corollary of Theorem 1.2. in \cite{Bor-Nebe-Involution} while the
main idea for the other orders is the following: if
$\sigma\in\Aut(C)$ is an automorphism of order $p\cdot q$, then, in
some cases, we can classify the possible sums
$F_{\sigma^q}(C)+F_{\sigma^p}(C)$. If no sum has minimum distance
greater than or equal to $24$, then an automorphism of this order
cannot occur. Note that these methods are a simplified version of
those in Section V of \cite{Borello6}.

\vspace{3mm}

\textbf{Automorphism of order 8:} Let $\sigma \in \Aut(C)$ be a
fixed point free automorphism of order $8$. Then $\sigma$ is of type
$8$-$(15;0)$. By Theorem 1.2. in \cite{Bor-Nebe-Involution}, $C$ is
a free $\F_2\langle\sigma^4\rangle$-module, so, by Chouinard's
Theorem \cite{chou}, $C$ is a free
$\F_2\langle\sigma\rangle$-module. This is impossible, since $8$
does not divide $60$.

\begin{theorem}
The automorphism group of a self-dual $[120,60,24]$ code does not
contain fixed point free elements of order $8$.
\end{theorem}

\vspace{3mm}

\textbf{Automorphism of order 30:} Let $\sigma \in \Aut(C)$ be of
order $30$. Then $\sigma$ is of type $30$-$(4;0)$. We can suppose,up
to equivalence, that
$$\sigma=(1,\ldots,30)(31,\ldots,60)(61,\ldots,90)(91,\ldots,120).$$
Let $\sigma_3:=\sigma^{10}$ and $\sigma_5:=\sigma^6$. Then
$\sigma_3$ is of type $3$-$(40;0)$ and $\sigma_5$ is of type
$5$-$(24;0)$. Since $\sigma$ is in the centralizer of both
$\sigma_3$ and $\sigma_5$ in $S_{120}$, it acts on
$\pi_{\sigma_3}(F_{\sigma_3}(C))$ as an automorphism, say
$\pi_{\sigma_3}(\sigma)$, of type $10$-$(4;0)$ and on
$\pi_{\sigma_5}(F_{\sigma_5}(C))$ as an automorphism, say
$\pi_{\sigma_5}(\sigma)$, of type $6$-$(4;0)$. Among the 16470
self-dual $[40,20,8]$ codes, only 28, say $D_1,\ldots,D_{28}$, have
an automorphism of this type, for a total of $69$ conjugacy classes.
So, up to a permutation in $C_{S_{40}}(\pi_{\sigma_3}(\sigma))$,
$\pi_{\sigma_3}(C)$ belong to a set, say $\mathcal{D}$, of $69$
elements. On the other hand, the extended binary Golay code has only
one conjugacy class of elements of type $6$-$(4;0)$. If $E_0$ is an
extended binary Golay code with automorphism
$\pi_{\sigma_5}(\sigma)$, then the orbit, say $\mathcal{E}$, of
$E_0$ under the action of $C_{S_{24}}(\pi_{\sigma_5}(\sigma))$ has
1296 elements. The code $\pi_{\sigma_5}(F_{\sigma_5}(C))$ belongs to
$\mathcal{E}$. With {\sc Magma} \cite{magma} we check that all the
codes in
$\mathcal{C}:=\{\pi_{\sigma_3}^{-1}(D)+\pi_{\sigma_5}^{-1}(E) \ | \
D\in \mathcal{D}, E\in\mathcal{E}\}$ have minimum distance less then
$24$. Since $F_{\sigma_3}(C)+F_{\sigma_5}(C)\subseteq C$ would have
to belong to $\mathcal{C}$, this implies the following result.

\begin{theorem}
The automorphism group of a self-dual $[120,60,24]$ code does not
contain elements of order $30$.
\end{theorem}

This theorem implies that also automorphism of order $60$ cannot
occur in $\Aut(C)$.

\vspace{3mm}

\textbf{Automorphism of order 57:} Let $\sigma \in \Aut(C)$ be of
order 57. Then $\sigma$ is of type $3\cdot 19$-$(2,0,2;0)$. We can
suppose,up to equivalence, that
$$\sigma=(1,\ldots,57)(58,\ldots,114)(115,116,117)(118,119,120).$$
Let $\sigma_3:=\sigma^{19}$ and $\sigma_{19}:=\sigma^3$. Then
$\sigma_3$ is of type $3$-$(40;0)$ and $\sigma_{19}$ is of type
$19$-$(6;6)$. Since $\sigma$ is in the centralizer of both
$\sigma_3$ and $\sigma_{19}$ in $S_{120}$, it acts on
$\pi_{\sigma_3}(F_{\sigma_3}(C))$ as an automorphism, say
$\pi_{\sigma_3}(\sigma)$, of type $19$-$(2;2)$ and on
$\pi_{\sigma_{19}}(F_{\sigma_{19}}(C))$ as an automorphism, say
$\pi_{\sigma_{19}}(\sigma)$, of type $3$-$(4;0)$. Among the 16470
self-dual $[40,20,8]$ codes, only 3, say $D_1, D_2$ and $D_{3}$,
have an automorphism of this type, for a total of $396$ conjugacy
classes. So, up to a permutation in
$C_{S_{40}}(\pi_{\sigma_3}(\sigma))$, $\pi_{\sigma_3}(C)$ belong to
a set, say $\mathcal{D}$, of $396$ elements. On the other hand, the
code $B_{12}$ has only one conjugacy class of elements of type
$3$-$(4;0)$. If $E_0$ is a $B_{12}$ code with automorphism
$\pi_{\sigma_{19}}(\sigma)$, then the orbit, say $\mathcal{E}$, of
$E_0$ under the action of $C_{S_{12}}(\pi_{\sigma_{19}}(\sigma))$
has  27 elements. The code $\pi_{\sigma_{19}}(F_{\sigma_{19}}(C))$
belongs to $\mathcal{E}$. With {\sc Magma} \cite{magma} we check
that all the codes in
$\mathcal{C}:=\{\pi_{\sigma_3}^{-1}(D)+\pi_{\sigma_{19}}^{-1}(E) \ |
\ D\in \mathcal{D}, E\in\mathcal{E}\}$ have minimum distance less
then $24$. Since $F_{\sigma_3}(C)+F_{\sigma_{19}}(C)\subseteq C$
would have to belong to $\mathcal{C}$, this implies the following
result.

\begin{theorem}
The automorphism group of a self-dual $[120,60,24]$ code does not
contain elements of order $57$.
\end{theorem}

\vspace{3mm}

\textbf{Other orders}: in the case of automorphisms of order $12$
(fixed point free), $15$, $20$ (fixed point free) and $115$ we do
not get a contradiction on the minimum distance, while in the case
of the automorphism of order $2\cdot p$, with $p$ prime, we cannot
use the method above, since we do not have a classification of the
fixed code by the automorphism of order $2$.

\subsection{Structure of the automorphism group in the fixed point free case}\label{ssfpf}

In this subsection we present a theorem on the structure of the
automorphism group of a self-dual $[120,60,24]$ code as in Section 6
of \cite{Bor-summary} for the self-dual $[72,36,16]$. Note that in
\cite{Bouy-24m} it is proved that involutions acting on extremal
codes of length $24m$ with $m>1$ are always fixed point free, except
for $m=5$, i.e. our case. It seems to be very difficult, although
very interesting, to exclude this exceptional case. Allowing fixed
points increases enormously the number of possible automorphism
groups and we cannot get nice results. Therefore we decided, as in
\cite{B-DC-W}, to restrict our attention to the fixed point free
case in order to get, at least under this hypothesis, a stronger
result.

\begin{theorem}
If all the involutions act fixed point freely, the automorphism
group $G$ of a self-dual $[120,60,24]$ code is trivial or isomorphic
to one of the following $64$ groups:

\vspace{-3mm}

$$\begin{tabular}{|m{0.9cm}|m{4.5cm}|| m{0.9cm}| m{6.5cm}|}
\hline
{\rm Order} & {\rm Groups} & {\rm Order} & {\rm Groups}\\
\hline \hline
$2$ & $C_2$ &  $23$ & $C_{23}$\\
\hline
$3$ & $C_3$ &  $24$ &  $Dic_{24}$, $S_3\times C_4$, $D_{12}$, $Dic_{12}\times C_2$, $C_3\rtimes D_4$, $C_6\times C_4$, $D_4\times C_3$, $S_4$,  $A_4\times C_2$, $D_6\times C_2$, $C_2\times C_2\times C_2\times C_3$\\
\hline
$4$ & $C_4$,  $C_2\times C_2$ & $30$ & $D_{15}$, $C_5\times S_3$ $C_3\times D_5$\\
\hline
$5$ & $C_5$ &  $38$ & $D_{19}$\\
\hline
$6$ & $C_6$, $S_3$ &  $40$ & $C_{20}\times C_2$, $D_5\times C_4$, $C_5\rtimes (C_4\times C_2)$, $D_{20}$, $C_5\rtimes D_4$, $D_4\times C_5$, $GA(1,5)\times C_2$, $D_5\times C_2\times C_2$, $C_2\times C_2\times C_2\times C_5$\\
\hline
$7$ & $C_7$ &    $56$ & $(C_2\times C_2\times C_2)\rtimes C_7$\\
\hline
$8$ & $C_4\times C_2$, $C_2\times C_2\times C_2$, $D_4$ &    $57$ & $C_{19}\rtimes C_3$\\
\hline
$10$ & $C_{10}$, $D_5$  & $60$ & $A_5$, $D_5\rtimes C_6$, $C_{15}\rtimes C_4$, $D_{15}\rtimes C_2$, $A_4\times C_5$\\
\hline $12$ & $C_{12}$, $C_6\times C_2$, $D_6$, $A_4$, $Dic_{12}$
& $114$ & $C_{19}\rtimes C_6$ \\
\hline
$15$ & $C_{15}$ & $115$ & $C_{115}$\\
\hline
$19$ & $C_{19}$ &  $120$ & $S_5$, $A_5\times C_2$, $S_4\times C_5$, $A_4\rtimes D_5$, $A_4\times D_5$\\
\hline
$20$ & $C_{20}$, $C_{10}\times C_2$, $D_{10}$, $Dic_{20}$, $GA(1,5)$ & & \\
\hline
\end{tabular}$$
\end{theorem}

\vspace{2mm}

\begin{proof}
All assertions about groups of order less than or equal to $552$
make use of the library SmallGroups of {\sc Magma} \cite{magma}.
Condition (e) of Remark \ref{stateofart} implies that the order of
$G$ is in $\{1, 2, 3, 4, 5, 6, 7, 8, 10, 12, 15, 19, 20, 23, 24,
30,$ $38,$ $40, 56, 57, 60, 114, 115, 120, 552, 2760\}$.\\ Moreover,
by Remark \ref{stateofart} and by the previous results, the order of
every element in $G$ is in $O:=\{1, 2, 3, 4, 5, 6, 7, 10, 12, 15,
19, 20, 23, 115\}$.

If $|G|=2760$, the there exists either one $23$-Sylow or $24$
$23$-Sylow subgroups. In the first case, the $23$-Sylow is normal
and its product with a 2-Sylow subgroup is a subgroup of order
$184$. All groups of order $184$ contain an element of order $46$.
In the second case, $G$ acts on the $24$ $23$-Sylow subgroups and
$G_H$ is of order $115$ for every $23$-Sylow subgroup $H$. Therefore
$G_H$ is cyclic. Let $K$ be the only subgroup of $G_H$ of order
$23$. This acts on $23$ groups (all except $H$) and so it has $23$
fixed points. Then $K$ is contained in $G_H$ for every $H$ and it
should be the unique group of order $23$ which is contained in $G_H$
for all $H$. This is not possible, since every $H$ is contained in
$G_H$, so we have a contradiction.

The quaternion group $Q_8$ cannot occur, again by Chouinard's
Theorem (see the proof above for the element of order $8$).

The group $C_5\rtimes (S_3\times C_4)$, of order $120$, is not
possible, since if $\sigma$ is the element of order $5$, then in the
automorphism group of  $\pi_\sigma(F_\sigma(C))$ (which is an
extended binary Golay code) there should be a subgroup isomorphic to
$S_3\times C_4$ acting fixed point freely, and this is not the case.

Finally, all the other groups are excluded by verifying that they
have elements of order which is not in $O$, or a subgroup isomorphic
to $Q_8$ or to $C_5\rtimes (S_3\times C_4)$.
\end{proof}

\begin{remark}
It would be interesting to exclude other non-abelian groups, using
methods similar to those in \cite{Borello-Volta-Nebe}, or elementary
abelian groups, using methods similar to those in \cite{Borello8}.
However, for a lack of classification of smaller codes, this seems
to be still computationally impossible. It would be also interesting
to get similar results without the hypothesis of the fixed point
free action, but this seems to make the number of possibilities grow
enormously. Finally, another direction of further research can be to
get a similar result for the extremal code of length $96$, which is
studied in \cite{B-W-Y, DC-W-96, Dontcheva}, but this is beyond the
aim of this paper.
\end{remark}

\section{Some necessary conditions for the existence of a self-dual extremal [120, 60, 24] code}\label{section-necess}

In this section we establish some necessary conditions for the
existence of an extremal $[120, 60, 24]$ code. Similar conditions
are given in \cite{Dougherty} and \cite{formally-dual} for an
extremal $[72, 36, 16]$ and $[96, 48, 20]$ code.

Let $C$ be a $[118,59,22]$ type I code. By Theorem
\ref{Gleason-theorem} (a) we have
\begin{align*} W_{C}(y)= & \sum^{14}_{j=0} a_{j}(1+y^{2})^{59-4j}(y^{2}(1-y^{2})^{2})^{j}=\\
=&
a_{0}+(59a_{0}+a_{1})y^{2}+(1711a_{0}+53a_{1}+a_{2})y^{4}+(32509a_{0}+1376a_{1}+\\
& 47a_{2}
       +a_{3})y^{6}+ (455126a_{0}+23320a_{1}+1077a_{2}+41a_{3}+ a_{4})y^{8}+\dots
\end{align*}
with $a_{j}\in \mathbb{Q}$ for $j=0,\ldots,14.$ Since the minimum
distance of $C$ is $22$, we get $a_{0}=1$, $a_{1}=-59$,
$a_{2}=1416$, $a_{3}=-17877$, $a_{4}=128679$, $a_{5}=-538375$,
$a_{6}=1291628$, $a_{7}=-1713124$,
$a_{8}=1187434$, $a_{9}=-400374$ and $a_{10}=0.$\\
Let $S$ be the shadow of $C$. Then by Theorem
\ref{polynomial-type-I} (a) we have {\small\begin{align*} W_{S}(y)=
& \frac{1}{33554432}a_{14}y^{3}+(-\frac{1}{524288}a_{13} -
\frac{7}{8388608}a_{14})y^{7}+(\frac{189}{16777216}a_{14}+\\
      & \frac{1}{8192}a_{12}+ \frac{13}{262144}a_{13}) y^{11}+\left(-\frac{819}{8388608}a_{14}-\frac{1}{128}a_{11}-\frac{325}{524288}a_{13}\right.\\
      & \left.-\frac{3}{1024}a_{12}\right)y^{15}+  \left(\frac{325}{65536}a_{13}+\frac{69}{2048}a_{12}+\frac{11}{64}a_{11}+\frac{20475}{33554432}a_{14}\right) y^{19}+ \\
      &  \left(-\frac{231}{128}a_{11}-\frac{7475}{262144}a_{13}-\frac{12285}{4194304}a_{14}-\frac{253}{1024}a_{12}+12811968\right)y^{23}+ \ldots \\
\end{align*}}
 Let $W_{S}(y)=\sum_{i=0}^{118}B_{i}y^{i}.$ By Theorem \ref{polynomial-type-I} (b) we have $B_{i} \in \{0, 1 \}$ for $i=3,7$ and at most one $B_{i}$ is nonzero for $i\leq 11$. Therefore there are four possibilities:
\begin{itemize}
\item If $B_{3}=1$, $B_{7}=0,$ $B_{11}=0$, then $a_{14} = 33554432$, $a_{13} = -14680064$, $a_{12} = 2867200 $. Since $B_{15}=-2576-\frac{1}{128}a_{11} \geq 0$, we have $B_{19}=44275+\frac{11}{164}a_{11} < 0$, a contradiction.
\item If $B_{3}=0,$ $B_{7}=1,$ $B_{11}=0$, then $a_{14} = 0$, $a_{13} = -524288$, $a_{12} = 212992$. Since $B_{15}= -299-\frac{1}{128}a_{11} \geq 0$, we have $B_{19}=4576+\frac{11}{164}a_{11} < 0.$  It is again a contradiction.
\item If $B_{3}=0,$ $B_{7}=0,$ $B_{11}\neq 0$, then $a_{14}=a_{13}=0$ and $a_{12}>0$. Therefore the system of inequalities  $B_{15}=\frac{1}{-128}a_{11}-\frac{3}{1024}a_{12}\geq 0$ and $B_{19}=\frac{69}{2048}a_{12}+\frac{11}{64}a_{11}\geq 0$ has no solutions, a contradiction.
\end{itemize}

  Hence $B_{3}=0$, $B_{7}=0$, $B_{11}=0$ and we obtain $a_{14}=a_{13}=a_{12}=a_{11}=0$. In conclusion the shadow has minimal distance $23$ and we can calculate the weight enumerators $W_{C}(y)$ and $W_{S}(y)$ (see Table \ref{polyI} and Table \ref{shadow-T}).

\begin{table}[h]
\caption{The weight enumerator of a self-dual $[118, 59, 22]$
code}\label{polyI}
  \begin{center}
  \setlength{\doublerulesep}{-1pt}
\begin{tabular}{c|c}
  \hline
    \hline
   $i$ & $A_{i}$ \\
  \hline
  22 96 & 1534767 \\
  24 94 &  25357020 \\
  26 92 &  323009424 \\
  28 90 &  3577030288 \\
  30 88 &  33041945820 \\
  32 86 &  255009210885 \\
  34 84 &  1660986238080 \\
  36 82 &  9190790517376 \\
  38 80 &  43420813336368 \\
  40 78 &  175902467952336 \\
  42 76 &  613510461769920 \\
  44 74 &  1848313759032000 \\
  46 72 &  4823479510074576 \\
  48 70 &  10929799315381752 \\
  50 68 &  21547310072116608 \\
  52 66 &  37017173713636224 \\
  54 64 &  55486969304739115 \\
  56 62 &  72637487089840296 \\
  58 60 &  83095867738716768 \\
\end{tabular}
\end{center}
\end{table}

\begin{table}[h]
 \caption{The weight enumerator of the shadow of a self-dual $[118, 59, 22]$ code}\label{shadow-T}
  \begin{center}

  \setlength{\doublerulesep}{-1pt}
\begin{tabular}{c|c}
 \hline
 \hline
  $i$ & $B_{i}$ \\
  \hline
  23 95 & 12811968 \\
  27 91 & 2201249408  \\
  31 87 & 187592982720  \\
  35 83 & 7972733942784  \\
  39 79 & 178129081470720  \\
  43 75 & 2168688143930880  \\
  47 71 & 14778320201079552  \\
  51 67 & 57459493525644288  \\
  55 63 & 129133310381938304\\
  59& 169008544553322240  \\
\end{tabular}
\end{center}
\end{table}

\begin{lemma} \label{design-shadow}  Let $C_{0}$ be the subcode of $C$ containing all codewords whose weights are multiples of 4. Then the supports of all vectors of a given weight in $C_{0}$ and in $C_{0}^{\perp}$ form a $3$-design.
\end{lemma}

\begin{proof} Since $C_{0}$ is a  $[118, 58, 24]$ code and $C_{0}^{\perp}=C \cup S$, then $C_{0}^{\perp}$ is a $[118, 60, 22]$ code. If $W_{C_{0}}(x,y)=\sum_{i=0}^{118}A_{i}x^{118-i}y^{i}$, then we have $|\{i\mid A_{i}\neq 0, \, 0< i \leq 115\}|=19\leq d(C_{0}^{\perp})-3$. Therefore by Theorem \ref{Assmus} the supports of the vectors of weight $i$ in $C_{0}^{\perp}$ and in $C_{0}$ form a $3$-design.
\end{proof}

\begin{question}
By Assmus-Mattson theorem the supports of the codewords of minimal
weight 22 in a self-dual $[118, 59, 22]$ code build a $3$-$(118, 22,
8885)$ design
 $\mathcal{D}$. Similarly the supports of the codewords of minimal weight 24 in a self-dual $[120, 60, 24]$ code build a $5$-$(120, 24, 8855)$ design. In \cite{CW-design} the authors showed that if $\mathcal{D}$ is a self-orthogonal $5$-$(120, 24, 8855)$ design, then $C(\mathcal{D})$ is a self-dual $[120, 60, d]$ code with $d=20$ or $d=24$. Unfortunately it was not possible to exclude the case $d=20$ and show that the code is extremal. A natural question, although very difficult, is: if $\mathcal{D}$ is a self-orthogonal $3$-$(118, 22, 8885)$ design, then is $C(\mathcal{D})$ a self-dual extremal $[118,59,22]$ code?
\end{question}

\begin{remark} Since the shadow of a self-dual $[118, 59, 22]$ code has minimal distance $23$, the existence of a self-dual
  $[120, 60, 24]$ code is equivalent to the existence of a self-dual $[118, 59, 22]$ code $C$ (In general by \cite{Rains1} the existence of an extremal code of length $24m$ is equivalent to the existence of a self-dual $[24m-2, 12m-1, 4m+2]$ code).
By Lemma \ref{design-shadow} the supports of the vectors of weight
$k$ in the shadow $S$ of $C$ form a
  $3$-$(118,k,\lambda)$ design. Therefore if there exists a $[118, 59, 22]$ code and its shadow has enumerator weight $W_{S}(y)=\sum_{k=0}^{118}B_{k}y^{k}$, then the coefficients of the shadow must satisfy the following condition of divisibility $\frac{(k-i)!}{k!}\frac{118!}{(118-i)!} \mid B_{k}$, since the terms $\lambda_{i}=B_{k}\frac{k!}{(k-i)!}\frac{(118-i)!}{118!}\in \mathbb{N}_{0} $. Actually, all the coefficients $B_{k}$ satisfy this condition (see Table \ref{shadow-condition}).
\end{remark}

\begin{table}[h]
  \caption{Parameters of the $3$-designs}\label{shadow-condition}
  \begin{center}

  \setlength{\doublerulesep}{-1pt}
{\scriptsize \begin{tabular}{c|c|c|c|c}
  \hline

  \hline
  $k$ & $\lambda_{0}=B_{k}$ & $\lambda_{1}$ & $\lambda_{2}$ & $\lambda_{3}=\lambda$ \\
  \hline
  23 & 12811968 & 2497248 & 469568 & 85008 \\
  27 & 2201249408 & 503675712 & 9327328 & 24122400 \\
  31 & 187592982720 & 49282902240 & 12636641600 & 3159160400 \\
  35 & 7972733942784 & 2364793966080 & 687205084160 & 195497998080 \\
  39 & 178129081470720 & 58873170994560 & 19121200835840 & 6099003714880 \\
  43 & 2168688143930880 & 790284662618880 & 283691930170880 & 100270423594880 \\
  47 & 14778320201079552 & 5886280080091008 & 2314263963112704 & 897774813276480 \\
  51 & 57459493525644288 & 24834187879727616 & 10612900803302400 & 4483035684153600 \\
  55 & 129133310381938304 & 60189254839039040 & 27779656079556480 & 12692429070831840 \\
  59 & 169008544553322240 & 84504272276661120 & 41891006769626880 & 20584374016109760 \\
  63 & 129133310381938304 & 68944055542899264 & 36534456783416704 & 19212085032658784 \\
  67 & 57459493525644288 & 32625305645916672 & 18404018569491456 & 10312596612215040 \\
  71 & 14778320201079552 & 8892040120988544 & 5320024004010240 & 3164497036868160 \\
  75 & 2168688143930880 & 1378403481312000 & 871810748864000 & 548639522992000 \\
  79 & 178129081470720 & 119255910476160 & 79503940317440 & 52774167279680 \\
  83 & 7972733942784 & 5607939976704 & 3930351094784 & 2744469298944 \\
  87 & 187592982720 & 138310080480 & 101663819840 & 74495040400 \\
  91 & 2201249408 & 1697573696 & 1305825920 & 1001883680 \\
  95 & 12811968 & 10314720 & 8287040 & 6643920 \\
  \hline
\end{tabular}}
\end{center}
\end{table}

Now we have the following necessary condition on the existence of an
extremal type II code of length $120$.

\begin{theorem} If no linear $[120, 60, 23]$ code with weight enumerator given in Table \ref{formally-self-dual-120} exists, then there exists no self-dual $[120, 60, 24 ]$ code.
\end{theorem}

\begin{proof} A \ self-dual \ $[120,60, 24]$ \ code \ has \ a \ child \ $C$, which is a self-dual $[112,59,22]$ code. By Lemma \ref{same-polynomial}, $W_{C_{1}}(y)=W_{C_{3}}(y)$. Therefore the code $C^\ast$ defined as in Lemma \ref{prop1} is formally self-dual, and the theorem follows (the weight enumerator of the code $C^\ast$ is given in Table \ref{formally-self-dual-120}. It is calculated thanks to  Lemma \ref{prop1}, knowing the weight enumerator of the $[118, 59, 22]$ type I code and of its shadow, which are given in Table \ref{polyI} and Table \ref{shadow-T} respectively).
\end{proof}

\begin{table}[h]
\caption{The weight enumerator of a formally self-dual $[120, 60,
23]$ code}\label{formally-self-dual-120}
\begin{center}
\setlength{\doublerulesep}{-1pt}
\begin{tabular}{r r|r r }
  \hline

  \hline
  Weight & Weight distribution & Weight & Weight distribution \\
  \hline
  0 & 1 &   61 & 84504272276661120  \\
  23 & 1534767 &   63 & 72637487089840296  \\
  24 &31763004 &  64 & 120053624495708267  \\
  25 &6405984 &   65& 64566655190969152 \\
  27 &323009424 &  67& 37017173713636224 \\
  28 & 4677654992&  68& 50277056834938752 \\
  29 & 1100624704&  69& 28729746762822144 \\
  31&33041945820 &  71& 10929799315381752  \\
  32&348805702245 &  72& 12212639610614352 \\
  33& 93796491360&  73& 7389160100539776 \\
  35& 1660986238080 &  75& 1848313759032000 \\
  36&13177157488768 &  76& 1697854533735360 \\
  37& 3986366971392&  77& 1084344071965440 \\
  39&43420813336368 &  79& 175902467952336 \\
  40&264967008687696 &  80& 132485354071728 \\
  41&89064540735360 &  81& 89064540735360 \\
  43&613510461769920 & 83& 9190790517376 \\
  44&2932657830997440 &  84& 5647353209472 \\
  45&1084344071965440 &  85& 3986366971392 \\
  47&4823479510074576 &  87& 255009210885 \\
 48&18318959415921528 &  88& 126838437180\\
 49& 7389160100539776&  89& 93796491360\\
 51& 21547310072116608&  91& 3577030288\\
 52&65746920476458368 &  92& 1423634128\\
 53&28729746762822144 &  93& 1100624704\\
 55&55486969304739115 &  95& 25357020\\
 56&137204142280809448 &  96& 7940751\\
 57&64566655190969152 &  97& 6405984\\
 59&83095867738716768 &  \\
 60&167600140015377888 &  \\
  \hline
 \end{tabular}
\end{center}
\end{table}

\begin{theorem}  If no self-dual doubly-even $[120, 60, 4]$ code with weight enumerator given in Table \ref{nachbar-120} exists, then there exists no self-dual $[120, 60, 24]$ code.
\end{theorem}
\begin{proof} Let $C$ be a self-dual $[120, 60, 24]$ code, $u\in \mathbb{F}_{2}^{120}$ of weight $\wt(u)=4$ and $D:=C \cap \langle u \rangle^\perp$. Since $D^\perp=C^\perp+ (\langle u\rangle^\perp)^\perp= C \oplus \langle u\rangle$, then $D \leq D^\perp$ and $\dim D^\perp=\dim(C+\langle u\rangle)=61$ i.e. $D$ is a self-orthogonal $[120, 59]$ code. If $N:=\langle D, u \rangle$, then $N= D \oplus \langle u\rangle$ and $\dim N=60$. Since $(d_1+u)\cdot (d_2+u)=0$, for all $d_1, d_2 \in D$, we have that $N$ is self-orthogonal and therefore it is self-dual. On the other hand, $\wt(d+u) \equiv 0 \mod 4$, because $d\cdot u=0$ for all $d \in D$. Therefore $N$ is a self-dual doubly-even $[120, 60, 4]$ code.
  Since $\dim(N \cap C)=59$, we have $N$ is neighbor of $C$. The code $N$ has only a vector of weight $4$. To determine the number of vectors of weight $20$ it is sufficient to calculate the number of vectors $w \in C$ of weight $24$ with $|\supp(w)\cap \supp(u)|=4$. Since the vectors of weight $24$ in $C$ form a $5$-design, then this number is equal to $\lambda_{4}=51359$ where $\lambda_{4}$ is the number of blocks incident with 4 different points. If $\sum A_{i}y^i$ is the weight enumerator of $N$, then $A_{0}=1$, $A_{4}=1$, $A_{8}=0$, $A_{12}=0$, $A_{16}=0$ and $A_{20}=51359$. Therefore by Theorem \ref{Gleason-theorem} we obtain the weight enumerator of the code $N$, which is given in Table \ref{nachbar-120}. This concludes the proof.
\end{proof}

\begin{table}[h]
\caption{The weight enumerator of a $[120, 60, 4]$ neighbor of an
extremal $[120,60, 24]$ code}\label{nachbar-120}
\begin{center}
\setlength{\doublerulesep}{-1pt}
\begin{tabular}{c|c}
  \hline
    \hline
   $i$ & $A_{i}$ \\
  \hline
  0 120 & 1 \\
  4  116 &  1 \\
  20 100 &  51359 \\
  24 96 &  43481179 \\
  28 92 &  6539254776 \\
  32 88 &   494044041905 \\
  36 84 &  19178964940125 \\
  40 80 &  400399951557816 \\
  44 76 &  4639015235035296 \\
  48 72 &   30526043817770504\\
  52 68 &   115980280893408771\\
  56 64 &   257259077150523955\\
  60    &   335272511326715600\\
\end{tabular}
\end{center}
\end{table}

\section{Acknowledgements}

The first author was partially supported by PEPS - Jeunes
Chercheur-e-s - 2017. The second author was financially supported by
the Fundaci\'on para la Promoci\'on de la Investigaci\' on y la
Tecnolog\'ia, in the frame of the project no. 201505. Part of this
work was done while he was in the University of Zurich supported by
the Swiss Confederation through the Swiss Government Excellence
Scholarship no. 2016.0873.


\begin{thebibliography}{00}

\bibitem{AssmusKey} E.F.~Assmus, Jr. and J.D.~Key, Designs and Their Codes, \textit{Cambridge University Press} (1992).

\bibitem{Assmus-paper} E.F.~Assmus, Jr. and H.F.~Mattson Jr., New 5-designs, \textit{ J. Combin. Theory} {\bf 6} (1969) 122-151.

\bibitem{Turyn} E.F.~Assmus, Jr., H.F.~Mattson, J.r., and R.J.~Turyn, Research to develop the algebraic theory of codes, Air Force
Cambridge Res. Labs., Bedford, MA, Report AFCRL-67-0365, (1967).


\bibitem{formally-dual} K.~Betsumiya, M.~Harada, Formally Self-Dual Codes Related to Type II Codes, \textit{ Appl. Algebra Eng. Commun. Comput.} {\bf 14}(2)(2003) 81-88.

\bibitem{sd40} K.~Betsumiya, M.~Harada, and A.~Munemasa, A complete classification of doubly even self-dual codes of length 40, \textit{Electron. J. Combin.} {\bf 19}(3)(2012):P18 (12 pp.).

\bibitem{Borello6} M.~Borello, The automorphism group of a self-dual [72, 36, 16] binary
code does not contain elements of order 6, \textit{IEEE Trans. Inf.
Theory}, {\bf 58}, no. 12, (2012), 7240-7245.

\bibitem{Bor-W-2p} M.~Borello and W.~Willems, Automorphisms of Order $2p$ in Binary Self-Dual
Extremal Codes of Length a Multiple of 24, \textit{IEEE Trans.
Inform. Theory}, {\bf 59}, No. 6 (2013) 3378-3383.

\bibitem{Borello-Volta-Nebe} M.~Borello, F.~Dalla~Volta, and G.~Nebe, The automorphism group of a
self-dual [72, 36, 16] code does not contain $S_3$, $A_4$, or $D_8$,
\textit{Adv. Math. Commun.,} {\bf 7}, no. 4, (2013) 503-510.

\bibitem{Borello8} M.~Borello, The automorphism group of a self-dual [72, 36, 16] code is
not an elementary Abelian group of order 8, \textit{Finite Fields
Their Appl.}, {\bf 25}, (2014) 1-7.

\bibitem{Bor-Nebe-Involution} M.~Borello and G.~Nebe, On involutions in extremal self-dual codes and the dual distance of semi self-dual codes, \textit{Finite Fields and Their Applications},
{\bf 33}, (2015) 80--89.

\bibitem{Bor-summary} M.~Borello, On automorphism groups of binary linear codes, \textit{Topics in finite fields,
Contemp. Math.}, \textbf{632}, Amer. Math. Soc., Providence, RI,
(2015) 29–-41.


\bibitem{magma}
    W.~Bosma, J.~Cannon, C.~Playoust, The \textsc{Magma} algebra system \textsc{I}: The user
    language, \textit{J. Symbol. Comput.}, {\bf 24} (1997) 235--265.

\bibitem{Bouy-24m} S.~Bouyuklieva, On the automorphisms of order 2 with fixed points for
the extremal self-dual codes of length 24m,\textit{Des., Codes,
Cryptograph.}, {\bf 25}, no. 1, (2002) 5-13.

\bibitem{B-DC-W} S.~Bouyuklieva, J.~de~la~Cruz, W.~Willems, On the automorphism group of a binary self-dual [120,60,24] code, \textit{AAECC}, {\bf 24}(3-4)(2013) 201-214.

\bibitem{B-W-Y} S.~Bouyuklieva, W.~Willems, N.~Yankov, On the automorphisms of order 15 for a binary self-dual $[96, 48, 20]$ code, \textit{Designs Codes and Cryptography}, {\bf 79}, (2016) 171-182.

\bibitem{chou} L.~Chouinard, Projectivity and relative projectivity over group rings, \textit{J. Pure Appl. Algebra} \textbf{7} (1976) 278–-302.

\bibitem{conway-sloane} J.H.~Conway and N.J.A.~Sloane, A New Upper Bound on the Minimal Distance of Self-Dual Codes, \textit{IEEE Trans. Inform. Theory}, {\bf 36}, No.
6 (1990).

\bibitem{C} J.~de~la~Cruz, On extremal self-dual codes of length 120, \textit{Des. Codes and Cryptography}, Vol. {\bf 75} (2015) 243-252.

\bibitem{CKW} J.~de~la~Cruz, M.~Kiermaier, A.~Wassermann, The automorphism group of an extremal $[120, 60, 24]$ code does not contain elements of order 29, \textit{Des. Codes and Cryptography} (2014) 1-10.

\bibitem{DC-W-96} J.~de~la~Cruz and W.~Willems, On Extremal Self-Dual Codes of Length 96, \textit{IEEE Trans. Inf. Theory}, {\bf 57}, no. 10 (2011), 6820-6823.

\bibitem{CW-design} J.~de~la~Cruz, W.~Willems, 5-Designs related to binary extremal self-dual codes of length 24m, \textit{Theory an applications of finite fields, Contemp. Mathe}, {\bf 579} (2012) 75-80.

\bibitem{Dontcheva} R.~Dontcheva, On the Doubly-Even Self-Dual
Codes of Length 96, \textit{IEEE Trans. Inform. Theory}, {\bf 48},
no. 2, (2002), 557-561.

\bibitem{Dougherty} S.T.~Dougherty, The Search for the $[24k, 12k, 4k+4]$ Extremal Type II Code, Scranton, PA 18510
November 29, 2006.

\bibitem{Gleason} A.M.~Gleason, Weight polynomials of codes and the MacWilliams identities, \textit{Actes Congr\'es Intern. de Math.}, Gauthier-Villars, Paris, {\bf 3} (1971) 211-215.

\bibitem{Huffman} W.C.~Huffman, Automorphisms of Codes with Applications
to Extremal Doubly Even Codes of Length $48$, \textit{IEEE Trans.
Inform. Theory} IT-28 (1982) 511-521.

\bibitem{Pless} G.T.~Kennedy and V.S.~Pless, A coding theoretic approach to extending designs,
 \textit{Discrete Appl. Math.} {\bf 142} (1995) 155-168.

\bibitem{Mallows} C.L.~Mallows, N.J.A.~Sloane,
An upper bound for self-dual codes, \textit{Information and Control}
{\bf 22} (1973) 188-200.

\bibitem{clasifi} V.~Pless, A classification of self-orthogonal
codes over $GF(2)$, \textit{Discr. Math.}, {\bf 3} (1972) 209-246.

\bibitem{Rains1} E.M.~Rains, Shadow Bounds for Self-Dual Codes, \textit{IEEE Trans. Inform. Theory}, {\bf 44} (1998) 134-139.

\bibitem{pregunta} N.J.A.~Sloane, Is there a $[72, 36]$, $d=16$ self-dual code?
\textit{IEEE Trans. Inform. Theory}, {\bf 19} (1973) 251.

\bibitem{Wasserman} R.~Yorgova, A.~Wassermann, Binary
self-dual codes with automorphisms of order $23$, \textit{Des. Codes
and Cryptography} {\bf 48} (2008) 155-164.

\bibitem{Zhang} S.~Zhang, On the nonexistence of extremal self-dual
codes, \textit{Discrete Appl. Math.} {\bf 91} (1999) 277-286.

\end{thebibliography}
\end{document}